\documentclass[conference]{IEEEtran}

\usepackage{cite}
\usepackage{graphicx}
\usepackage{url}
\usepackage{amsmath}
\usepackage{latexsym,amssymb}

\newcommand{\BEQA}{\begin{eqnarray}}

\newcommand{\EEQA}{\end{eqnarray}}

\newcommand{\figSpace}{\vspace{-0.2in}}
\newcommand{\figCaptionSpace}{\vspace{-0.1in}}

\newcommand{\Capprox}{\tilde{C}}
\newcommand{\D}[2]{\frac{\partial #1}{\partial #2}}

\newtheorem{theorem}{Theorem}

\newtheorem{definition}{Definition}

\begin{document}


\title{Multipath Wireless Network Coding: A Population Game Perspective }

\author{\IEEEauthorblockN{Vinith Reddy\IEEEauthorrefmark{1},
Srinivas Shakkottai\IEEEauthorrefmark{1},
Alex Sprintson\IEEEauthorrefmark{1} and
Natarajan Gautam\IEEEauthorrefmark{2}}
\IEEEauthorblockA{\IEEEauthorrefmark{1}Dept. of ECE, Texas A\&M University}
\IEEEauthorblockA{\IEEEauthorrefmark{2}Dept. of ISE, Texas A\&M University}
\IEEEauthorblockA{Email: \{vinith\_reddy, sshakkot, spalex, gautam\}@tamu.edu}
} 

\maketitle

\begin{abstract} 
We consider wireless networks in which multiple paths are available between each source and destination. We allow each source to split traffic among all of its available paths, and ask the question: how do we attain the lowest possible number of transmissions to support a given traffic matrix? Traffic bound in opposite directions over two wireless hops can utilize the ``reverse carpooling'' advantage of network coding in order to decrease the number of transmissions used. We call such coded hops as ``hyper-links''. With the reverse carpooling technique longer paths might be cheaper than shorter ones. However, there is a prisoners dilemma type situation among sources -- the network coding advantage is realized only if there is traffic in both directions of a shared path. We develop a two-level distributed control scheme that decouples user choices from each other by declaring a hyper-link capacity, allowing sources to split their traffic selfishly in a distributed fashion, and then changing the hyper-link capacity based on user actions. We show that such a controller is stable, and verify our analytical insights by simulation.
\end{abstract}

\section{Introduction}\label{sec:intro}

There has recently been significant interest in multihop wireless networks, both as a means for basic Internet access, as well as for building specialized sensor networks.  However, limited wireless spectrum together with interference and fading pose significant challenges for network designers. The technique of network coding has the potential to improve the throughput and reliability of multihop wireless networks by taking advantage of the broadcast nature of wireless medium.

For example, consider a wireless network coding scheme depicted in Figure~\ref{fig:network}(a). In this example, two wireless nodes need to exchange packets $x_1$ and $x_2$ through a relay node.  On the one hand, a simple \emph{store-and-forward} approach needs four transmissions.  On the other hand, the network coding approach uses a  \emph{store-code-and-forward} approach in which the two packets from the clients are combined by means of an XOR operation at the relay and broadcast to both clients simultaneously.  The clients can then decode this coded packet to obtain the packets they need.
\begin{figure}[h]
\begin {center}
\includegraphics[width=1.8in]{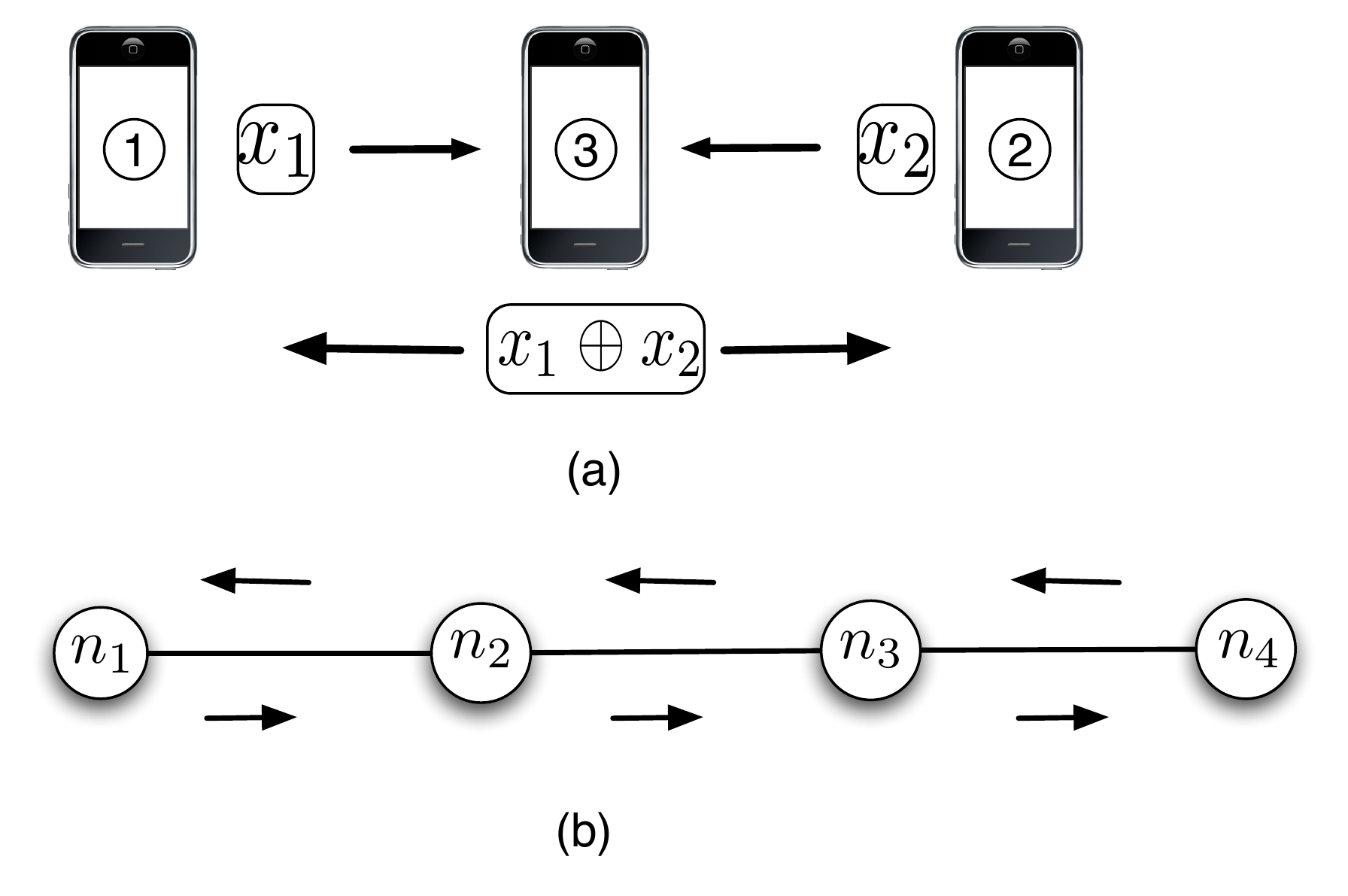}
\caption{(a) Wireless Network Coding (b) Reverse carpooling.
\label{fig:network}}
\vspace{-0.2in}
\end{center}
\end{figure}

Design and analysis of efficient network coding schemes for wireless networks have recently attracted a significant interest from the research community. Katti \emph{et al.} \cite{Our-Sigcomm06} presented a practical network coding architecture, referred to as \emph{COPE}, that implements the above idea while also making use of overheard packets to aid in decoding.  Experimental results shown in \cite{Our-Sigcomm06} indicate that the network coding technique can result in a significant improvement in the network throughput.

Effros \emph{et al.} \cite{1633782} introduced the strategy of \emph{reverse carpooling} that allows two information flows traveling in opposite directions to share a path. Figure~\ref{fig:network}(b) shows an example of two connections, from $n_1$ to $n_4$ and from $n_4$ to $n_1$ that share a common path $(n_1,n_2,n_3,n_4)$. The wireless network coding approach results in a significant (up to 50\%) reduction in the number of transmissions for two connections that use reverse carpooling. In particular, once the first connection is established, the second connection (of the same rate) can be established in the opposite direction with little additional cost.



The key challenge in the design of network coding schemes is to maximize the number of \emph{coding opportunities}, where a coding opportunity refers to an event in which at least one transmission can be saved by transmitting a combination of the packets. Insufficient number of coding opportunities may affect the performance of a network coding scheme and is one of the major barriers in realizing the coding advantage. Accordingly, the goal of this paper is to design, analyze, and validate network mechanisms and protocols that improve the performance of the network coding schemes through increasing the number of coding opportunities.

Consider the scenario depicted in Figure \ref{fig:codingOpportunities}. We have three sources of traffic, each of which is aware of two paths leading to its destination.  For example, Source $3$ (positioned at $n_5$) can send packets to its destination (located at $n_1$) at rates $x^1_3$ and $x^2_3$ on its two available paths.  We consider the cost metric of the system to be the number of transmissions required to support a given traffic matrix.  Under the current channel conditions, suppose that it is cheaper for Source $3$ to send all its traffic on path $(n_5,n_7,n_1)$.  However, notice that there is an opportunity for reverse carpooling on a subpath $(n_1,n_2,n_5)$.  With this scheme, node $n_2$ will broadcast coded packets to nodes $n_1$ and $n_5$. We refer to this broadcast link as a \emph{hyper-link}.  Although path $(n_5,n_2,n_1)$ is more expensive for Source $3$, if there is traffic from Source $1$ (located at $n_1$) that overlaps with it at $n_2$, it might actually be the case that the lowest cost traffic split in the system would entail that Source $3$ should use the hyper-link and send all its traffic using path $(n_5,n_2,n_1)$, while Source $1$ follows suit by using its path $(n_1,n_2,n_5,n_4)$.  

\begin{figure}[tb]
\vspace{-0.1in}
\begin{center}
\includegraphics[width=2.5in]{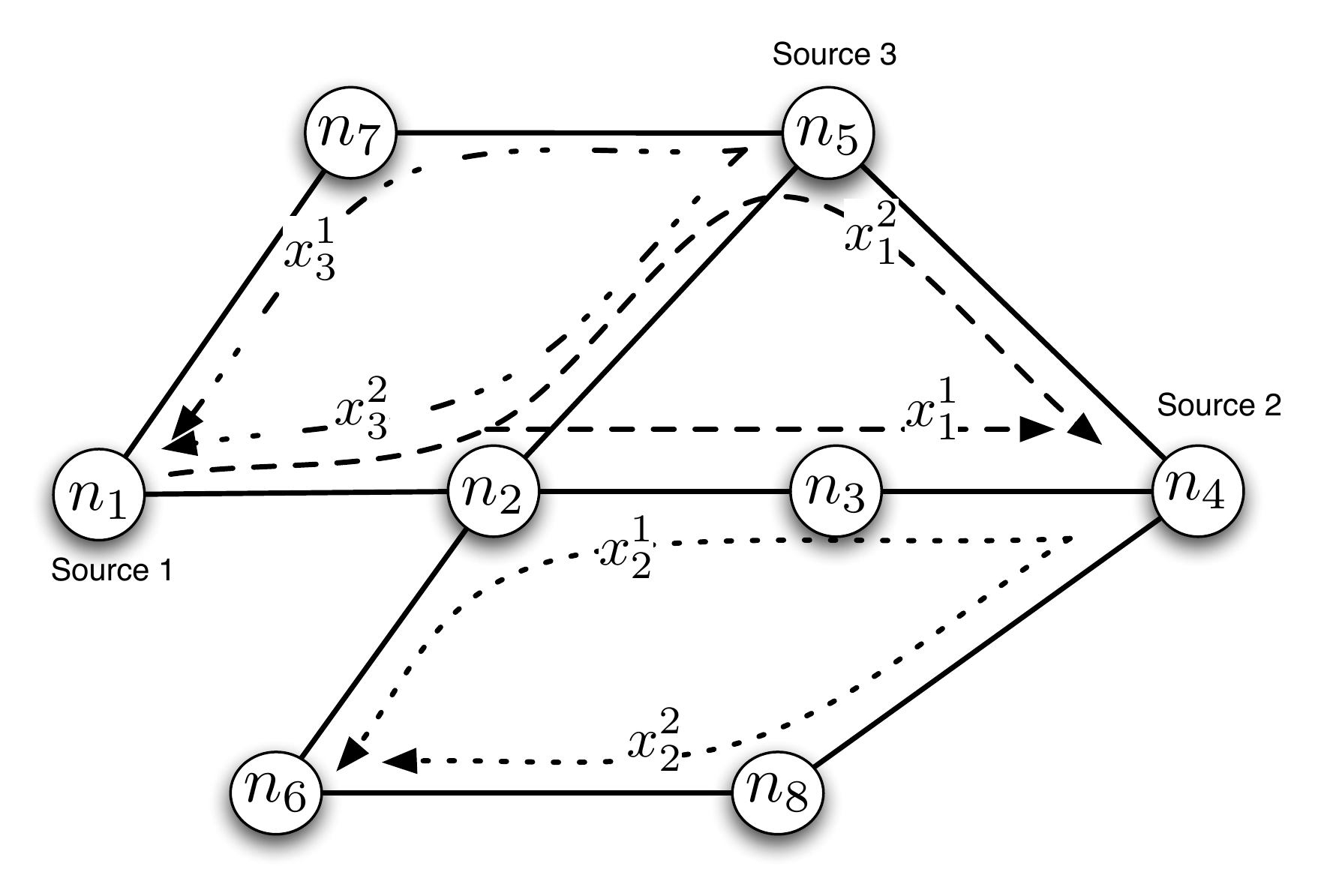}
\figCaptionSpace
\caption{Increasing Coding Opportunities}
\label{fig:codingOpportunities}
\end{center}
\figSpace
\end{figure}

However, we immediately see that there is a prisoners' dilemma situation here -- savings can only be obtained if there is sufficient bi-directional traffic on  $(n_1,n_2,n_5)$.  The first mover in this case is clearly at a disadvantage as it essentially creates the route that others can piggyback upon (in a reverse direction).  Our challenge in this paper is to design and validate a distributed control scheme that provides an incentive to use reverse carpooling, eliminates the first-mover disadvantage, and hence allows the system to attain the state of lowest possible cost to support its traffic.

\subsection{Related Work}\label{sec:related}

Network coding research was initiated by a seminal paper by Ahlswede \emph{et al.} \cite{ACLY00} and since then attracted a significant interest from the research community. Many initial works on the network coding technique focused on establishing \emph{multicast} connections between a fixed source and a set of terminal nodes. Li \emph{et al.} \cite{LYC03} showed that  the maximum rate of a multicast connection is equal to the minimum capacity of a cut that separates the source and any terminal.  In a subsequent work, Koetter and M\'{e}dard \cite{KM03} developed an algebraic framework for network coding and investigated linear network codes for directed graphs with cycles.

Network coding technique for wireless networks has been considered by Katabi \emph{et al.} \cite{Our-Sigcomm06}. The proposed architecture, referred to as COPE, contains a special network coding layer between the IP and MAC layers. In \cite{1282400} Chachulski  \emph{et al.} proposed an opportunistic routing protocol, referred to as MORE, that randomly mixes packets that belong to the same flow before forwarding them to the next hop.  Sagduyu and Ephremides \cite{4439862} focused on the applications of network coding in simple path topologies (referred to in \cite{4439862} as \emph{tandem} networks) and formulate a related cross-layer optimization problems.

Closest to our problem are \cite{das2008context,marden2009price}.   Das \emph{et al.} \cite{das2008context} propose a new framework called ``context based routing'' in multihop wireless networks that enables sources to choose routes that increase coding opportunities.  They propose a heuristic algorithm that measures the imbalance between flows in opposite directions, and if this imbalance is greater than $25\%$, provides a discount of $25\%$ to the smaller flow.  This has the effect of incentivizing equal bidirectional flows, resulting in multiple coding opportunities.  Our objective is similar, but we develop iterated distributed decision making that trades off a potential increase in cost of longer paths, with the potential cost reduction due to enhanced coding opportunities.  Marden \emph{et al.} \cite{marden2009price} consider a similar problem to ours, but unlike our focus on how to align user incentives, their focus is on the efficiency loss of the Nash equilibrium attained.  Our objective is to design an incentive structure that would naturally result in the system converging to the lowest cost state.

\subsection{Main Results}\label{sec:results}
The key contribution of this research is a distributed two-level control 
scheme that would iteratively try to lead the sources to discover the appropriate 
splits for their traffic among multiple paths. On one level are the
sources that selfishly choose to split their traffic across available 
multiple paths with costs and maximum capacities (set by the hyperlinks)
on each. On the other level,
the hyperlink nodes choose maximum capacities for paths that share that node
as a result of the sources' decisions as well as opportunities for network
coding.  Note by splitting up the dynamics in this fashion, our algorithm is a relaxation of the original cost-minimization problem.
The iteration process continues until the entire network has reached local minimum which, since our formulation is convex, is also the socially optimal solution.
We show that this process is asymptotically stable. We illustrate our approach as well as the quality of solution
using numerical experiments. The experiments indicate that: the convergence is fast; the costs are reduced
significantly upon using network coding; more expensive paths before 
network coding became cheaper and shortest paths were not necessarily
optimal.  Thus, the iterative algorithm that we develop from the relaxed formulation performs well in practice.  

\section{System Overview}\label{sec:overview}

Our objective is to design a distributed multi-path network coding system for multiple unicast flows traversing a shared wireless network.  We assume that the schedule of wireless links given to us (for example, using CSMA), and hence abstract out the interference between links.  We model the communication network as a graph $G(V,E)$, where $V$ is the set of network nodes and $E$ is the set of wireless links. For each link $(n_i,n_j)\in E$, where $(n_i,n_j)\in V$, there exists a wireless channel that allows node $n_i$ to transmit information to node $n_j$.  Each link $(n_i,n_j)$ is associated with a cost $\alpha_{ij}$. The value of $\alpha_{ij}$ captures the cost (number of transmissions) of transmitting information at unit rate from $n_i$ to $n_j$.  Due to a broadcast nature of the wireless channels, the node $n_i$ can transmit to two neighbors $n_j$ and $n_k$ simultaneously at a cost $\max\{\alpha_{ij},\alpha_{ik}\}$. 

We assume that the network supports flows $\{1,2,\dots,\}$, where each flow is associated with a source and destination node.  Each flow $i$ is also associated with several paths $\{P_i^1,P_i^2,\dots\}$ that connect its source and destination nodes.  Our goal is to build a distributed traffic management scheme in which the source node of each flow $i$ can split its traffic, $x_i$, among multiple different paths, so as to reduce the \emph{total number of transmissions, per unit rate,} required to support a given traffic.  Note that on some of these paths there might be a possibility of network coding. 

For example, consider the network depicted on Figure~\ref{fig:codingOpportunities}. The network supports three flows: (i) flow $1$ from $n_1$ to $n_4$, (ii) flow $2$ from $n_4$ to $n_6$, and (iii) flow  $3$ from $n_5$ to $n_1$.  We denote by $x_i$ the traffic associated with flow $i$, $1\leq i\leq 3$. Suppose that the packets that belong to flow $1$ can be sent over two paths $(n_1,n_2,n_3,n_4)$ and $(n_1,n_2,n_5,n_4)$. We denote these paths by $P_1^1$ and $P_1^2$. The traffic split on paths $P_1^1$ and $P_1^2$ is given by $x_1^1$, $x_1^2$, respectively, such that $x_1^1 + x_1^2 = x_1$. Similarly, flow $2$ can be sent over two paths $P_2^1 = (n_4,n_3,n_2,n_6)$ and $P_2^2 = (n_4,n_8,n_6)$ at rates $x_2^1$ and $x_2^2$, such that $x_2^1 + x_2^2 = x_2$. Finally, flow $3$ can be sent over two paths $P_3^1 = (n_5,n_7,n_1)$ and $P_3^2 = (n_5,n_2,n_1)$, at rates $x_3^1$ and $x_3^2$, with sum $x_3$. 

Note that path $P_1^2 = (n_1,n_2,n_5,n_4)$ of flow $1$ and path $P_3^2 = (n_5,n_2,n_1)$ of flow $3$  share two links $(n_1,n_2)$ and $(n_2,n_5)$ in the opposite directions. Thus, the packets sent along these two paths can benefit from reverse carpooling. Specifically, node $n_2$ can combine packets of flow $1$ received from node $n_1$ and packets of flow $3$ received from node $n_5$. Similarly, node $n_3$ can combine packets of flow $1$ received from node $n_2$ and packets of flow $2$ received from node $n_4$. Note that the cost saving at node $n_2$ is proportional to $\min\{x_1^2,x_3^2\}$, while the saving at node $n_3$ is proportional to $\min\{x_1^1,x_2^1\}$.  Note that our model is not restricted to reverse carpooling type XOR coding alone. Other types of XOR coding schemes like COPE\cite{cope06}, which uses ``opportunistic listening'' can also be used.

The cost at node $n_2$ when coding is enabled is
\BEQA\label{eqn:codingCost}
C_{n_2}(x_1^2,x_3^2) =& \max\{\alpha_{21},\alpha_{25}\}\min\{x_1^2,x_3^2\} \\
&+ \alpha_{25}(x_1^2 - \min\{x_1^2,x_3^2\}) \nonumber\\
&+ \alpha_{21}(x_3^2 - \min\{x_1^2,x_3^2\}). \nonumber
\EEQA
Here, the first term on the right is the cost incurred due to coding at node $n_2.$   This is because a coded packet from $n_2$ is broadcast to both destination nodes, $n_1$ and $n_5$, and so the cost per unit rate is $\max\{\alpha_{21},\alpha_{25}\}$. The second and third term are ``overflow'' terms. Since  its is possible that $x_1^2\neq x_3^2,$ the remaining flow of the larger (that cannot be encoded because of the lack of flow in the opposite direction) is sent without coding at the regular link cost. 

The cost at node $n_2$, given by (\ref{eqn:codingCost}), can be re-written as shown below:
\BEQA
C_{n_2}(x_1^2,x_3^2) &= \alpha_{25}x_1^2 &+ \alpha_{21}x_2^2 + \Big\{ \max\{\alpha_{21},\alpha_{25}\} \nonumber \\
&&-(\alpha_{21} + \alpha_{25}) \Big\}\min\{x_1^2,x_3^2\}. \nonumber
\EEQA
Using the fact that $\max\{x_1,x_2\} + \min\{x_1,x_2\} = x_1 + x_2$, we obtain
\BEQA\label{eqn:savings}
C_{n_2}(x_1^2,x_3^2) &=&  \alpha_{25}x_1^2 + \alpha_{21}x_2^2\\
 &-& \min\{\alpha_{21},\alpha_{25}\}\min\{x_1^2,x_3^2\}. \nonumber
\EEQA
The above equation can be interpreted as the cost at node $n_2$ without coding minus the savings obtained when coding is used. Thus, the cost saved at node $n_2$ due to network coding is $\min\{\alpha_{21},\alpha_{25}\}\min\{x_1^2,x_3^2\}$ . Similarly, for node $n_3$ the cost saved is $\min\{\alpha_{32},\alpha_{34}\}\min\{x_1^1,x_2^1\}.$
The total system cost can be expressed as:
\BEQA\label{equ:energy}
C(X) = \sum_{i=1}^3\sum_{j=1}^2\beta_i^jx_i^j\ -&\min\{\alpha_{21},\alpha_{25}\}\min\{x_1^2,x_3^2\} \\
-&\min\{\alpha_{32},\alpha_{34}\}\min\{x_1^1,x_2^1\}, \nonumber
\EEQA
where $X=\{x^1_1,x^2_1,x_2^1,x_2^2,x_3^1,x_3^2\}$ is the state of the system and $\beta_i^j$ is the uncoded path cost (equal to the sum of the link costs on the path) $j$ used by flow $i$. For example, $\beta^1_1 = \alpha_{12}+ \alpha_{23} + \alpha_{34}$, for path $P_1^1 = (n_1,n_2,n_3,n_4)$. Thus, the first term on the right in (\ref{equ:energy}) is the total cost of the system without any coding, while the second and third terms are the savings obtained by coding at nodes $n_2$ and $n_3.$

In this paper, we consider the problem of minimizing total cost, given the traffic matrix.  The problem poses major challenges due to the need to achieve a certain degree of coordination among the flows. For example, for the network depicted in Figure~\ref{fig:codingOpportunities}, increasing of the value of $x_3^2$ (the decision made by node $n_5$) will result in a system-wide cost reduction only if it is accompanied by the increase in the value of $x_1^2.$

\section{Hyper-links and System Cost}\label{sec:hyper}

In order to decouple the decisions of flows, we introduce the idea of a \emph{hyper-link} whose capacity can be controlled independently of the flows that use it.
\begin{definition}\label{def:hyper}
A \emph{hyper-link} is a broadcast-link composed of three nodes and two flows.  A hyper-link $n_k[(i,p,n_i),(j,q,n_j))]$ at node $n_k$ can encode packets belonging to flow $i$ (sending packets on path $p$) with flow $j$ (sending packets on path $q$). Here, nodes $n_i$ and $n_j$ are the next-hop neighbors of $n_k$; for flow $i$ along path $p$ and for flow $j$ along path $q$, respectively. 
\end{definition}
For each hyper-link $n_k[(i,p,n_i),(j,q,n_j))]$, we introduce a new decision variable $y_k$ that denotes the capacity of the hyper-link. This formulation helps us to decouple the coordination between individual flows.  We restrict the total coded (broadcast) traffic between the two flows at node $n_k$ to be at-most equal to the hyper-link capacity $y_k$.  Any remaining flow is sent without coding.  Referring to Figure~\ref{fig:codingOpportunities}, there exists a hyper-link $h_1 = n_2[(1,P_1^2,n_5),(3,P_3^2,n_1)]$, where the source node $n_2$ can encode packets of flow $f_1^2$ (flow along path $P_1^2$), destined to node $n_3$, with packets of flow $f_3^2$ (flow along path $P_3^2)$, destined to node $n_1$. Similarly, there exists a hyper-link 
$h_2 = n_3[(1,P_1^1,n_4),(2,P_2^1,n_2)]$, where the source node $n_3$ can encode packets of flow $f_1^1$, destined to node $n_4$, with packets of flow $f_2^1$ destined to node $n_2$. Let the hyper-link capacities be defined as $y_2$ and $y_3$ respectively.
The total cost of transmission on hyper-link  $h_1 = n_2[(1,P_1^2,n_5),(3,P_3^2,n_1)]$ of capacity $y_2$ is given by
\BEQA
C(h_1) &= \max\{\alpha_{25},\alpha_{21}\}y_2 +\label{eqn:hlCodingCost}\\
&  \alpha_{25}(x_1^2 - \min\{x_1^2,y_2\} + \nonumber \\
&  \alpha_{21}(x_3^2 - \min\{x_3^2,y_2\}, \nonumber
\EEQA
where the first term on the right is the cost of sending traffic on the hyper-link. Note that we have to bear this cost, regardless of whether or not there is enough bidirectional flow to be sent on the hyper-link.  This relaxation could potentially increase the total cost of the system.  However, as we will see in Section \ref{sec:hlCtrl}, we can design a hyper-link capacity controller which would adjust the hyper-link capacities periodically to minimize cost.  As before, the ``overflow'' packets are sent without coding, and the cost incurred in doing so is given by the latter two terms.

The cost at node $n_2$, given by (\ref{eqn:hlCodingCost}), can be re-written as:
\BEQA\nonumber
C(h_1) =& \alpha_{25}x_1^2 + \alpha_{21}x_3^2 - T(h_1)\mbox{, where } \hspace{0.5in}\\
T(h_1) =& \alpha_{25}\min\{x_1^2,y_2\} + \alpha_{21}\min\{x_3^2,y_2\} \nonumber\\
 &-\ \max\{\alpha_{25},\alpha_{21}\}y_2\nonumber
\EEQA
Recall that the first two cost terms are the total cost at node $n_2$ when coding is disabled. The remaining cost, $T(h_1),$ can be thought of as the \emph{rebate} obtained by using hyper-link $h_1 = n_2[(1,P_1^2,n_5),(3,P_3^2,n_1)]$. Note that the rebate could be \emph{negative} (hence adding to the total cost), which might happen when one of the flow's rate is $0$ and the other flow's rate is less than the hyper-link capacity. 

Thus, the modified cost function when the system is in state $(X,Y)$ is given by
\BEQA\label{equ:energy1}
C(X,Y) = \sum_{i=1}^3\sum_{j=1}^2\beta_i^jx_i^j - (T(h_1) + T(h_2)),
\EEQA
where $X=\{x^1_1,x^2_1,x_2^1,x_2^2,x_3^1,x_3^2\}$,  $Y =\{y_2,y_3\}$. $T(h_1)$ and  $T(h_2)$ are the rebates obtained by using hyper-link $h_1 = n_2[(1,P_1^2,n_5),(3,P_3^2,n_1)]$ and $h_2 = n_3[(1,P_1^1,n_4),(2,P_2^1,n_2)])$, respectively.
 
In general, the total system cost in terms of number of transmissions required to support a given traffic load, when the state of the system is $(X,Y)$, is:
\BEQA\label{eqn:cost}
C(X,Y) =& \mbox{Total system cost without coding} \nonumber  \\
       -& \mbox{Total rebate of all the hyper-links}
\EEQA

We focus on minimizing this total cost. To this end, we relax the problem into two sub-problems--that of traffic splitting by sources, and that of hyper-link capacity selection:  
\begin{enumerate}
\item {\bf Traffic Splitting:} In this phase, the source node of each flow splits its traffic among the different options, for a given hyper-link state $Y$. The options available to each flow are called \emph{hyper-paths}, where each such hyper-path contains zero or more hyper-links.  We model this phase as a potential game; the background needed is presented in Section \ref{sec:background}. 
Details of our game model and the payoffs used are covered in Section~\ref{sec:mpncGame}.
\item {\bf Hyper-Link Capacity Control:} In this phase, we adjust the hyper-link capacities in order to minimize the total cost. We use a simple gradient descent controller to attain minimum cost.  In this phase it is assumed the sources attain Wardrop equilibrium instantaneously. Further details on the type of controller used and the convergence properties are covered in Section~\ref{sec:hlCtrl}.
\end{enumerate}
We call our controller as \emph{Decoupled Dynamics}.  The two phases operate at different time scales. Traffic splitting is done at every \emph{small} time scale and the hyper-link capacity control is done at every \emph{large} time scale. Thus, sources attain equilibrium for a given hyper-link capacities, then the hyper-link capacities are adjusted, and this in turn forces the sources to change their splits. This process continues until the source splits and hyper-link capacities converge.

\section{Background: Potential Games}\label{sec:background}

Below we review some game-theoretic ideas that will be used in this paper.  Detailed discussion may be found in \cite{Sandholm01}. A \emph{population game} $\mathcal{G}$, with $F$ non-atomic populations of players is defined by a mass and a strategy set for each population and a payoff function for each strategy.  By a non-atomic population, we mean that the contribution of each member of the population is infinitesimal. We denote the set of populations by $\mathcal{F}=\{1, ... ,F\}$, where $F\geq 1$.  The population $i$  has mass $x_i$. The set of strategies for population $i$ is denoted $\mathcal{S}_i = \{1, ... , S_i \}$.  These strategies
can be thought of as the actions that members of $i$ could possibly take. A particular strategy distribution is the way the population $i$ partitions itself into the different actions available, i.e., a strategy distribution for $i$ is vector of the form $\vec{x}_i=\{x^1_i, x^2_i,...x^{S_i}_i\}$, where
$\sum_{p=1}^{S_i}{x^p_i}=x_i$. The set of strategy distributions of a population $i\in \mathcal{F}$, is denoted by $X_i =\{\vec{x}_i\in \mathbb{R}^{S_i}_+:\sum_{p=1}^{S_i}{x^p_i}=x_i\}$.   We denote the vector of strategy distributions being used by the entire population by ${\bf X}=\{\vec{x}_1,\vec{x}_2, ... , \vec{x}_F\}$, where $\vec{x}_i\in X_i$.  The vector ${\bf X}$ can be thought of as the state of the system. Let the  space of all strategy distributions be $\mathcal{X}$.

The marginal payoff function (per unit mass) obtained from strategy $p\in \mathcal{S}_i$ by users of class $i$, when the state of the system is ${\bf X}$ is denoted by $F^p_i({\bf X}) \in \mathbb{R}$ and is assumed to be continuous and differentiable. Note that the payoffs to a strategy in population $i$ can depend
on the strategy distribution within population $i$ itself.  The total payoff to users of class $i$ is then given by $\sum_{p=1}^{S_i}F^p_i({\bf X})x^p_i$, where we assume linearity for exposition. 

\emph{Potential games} are a type of population games, that have a specific structure on the cost function. The idea behind potential games is to identify a scalar function that represents the ``energy'' of the system (exactly like a Lyapunov function \cite{khalil96}), which is called the \emph{potential function}. All information regarding the payoffs obtained by users of a population class can be captured in the potential function.
\begin{definition}\label{defn:potentialGame}
Let $\mathcal{G}$ be a population game with payoff function (per unit mass) $F:\mathcal{X} \rightarrow \mathcal{R}^F$. $\mathcal{G}$ is called a \emph{Potential Game} if there exists a continuously differentiable function $\mathcal{T}:\mathcal{X} \rightarrow \mathcal{R}$ such that 
\BEQA
\D{\mathcal{T}}{x^p_i}(X) = F^p_i(X)
\EEQA
$\forall i \in \mathcal{F}$ and $p \in \mathcal{S}_i$, where $X \in \mathcal{X}$ is the state of the system. 
The function $\mathcal{T}$ is called the \emph{potential function} for game $\mathcal{G}$. 
\end{definition} 

Next, we define the concept of equilibrium in population games. A commonly used concept in non-cooperative games in the context of infinitesimal players, is the Wardrop equilibrium \cite{Wardrop52}. Consider any strategy distribution $\vec{x}_i=[x^1_i,...,x^{S_i}_i]$.  There would be some elements which are non-zero
and others which are zero.  We call the strategies corresponding to the non-zero elements as the \emph{strategies used by
population $i$}.
\begin{definition}\label{defn:wardrop}
A state ${\bf \hat{X}}$ is a Wardrop equilibrium if for any population $i \in \mathcal{F}$, all strategies being used by the members of $i$ yield the same marginal payoff to each member of $i$, whereas the marginal payoff that would be obtained is lower for all strategies not used by population $i$.

Let $\mathcal{\hat{S}}_i\subset \mathcal{S}_i$ be the set of all strategies used by population $i$ in a strategy
distribution ${\bf \hat{X}}$.  A Wardrop equilibrium ${\bf \hat{X}}$ is then characterized by the following relation:
\BEQA
F_i^s( {\bf \hat{X}} ) \geq F_i^{s'} ( {\bf \hat{X}})\ \ \forall s\in \mathcal{\hat{S}}_i\mbox{ and } s'\in \mathcal{S}_i
\nonumber
\EEQA
\end{definition}

 The above concept refers to an \emph{equilibrium condition}; the question arises as to how the system actually arrives at such a state.  
A commonly used kind of population dynamics is \emph{Brown-von Neumann-Nash (BNN) Dynamics}~\cite{broneu50}.  The dynamics are described as follows: 
\BEQA\label{eqn:origbnn}
\dot{x}_i^p  &=& \left(x_i\gamma_i^p - x^p_i\sum_{j=1}^{S_i}\gamma_i^j \right) \\
\mbox{ where, } 
\gamma_i^p &=& \max\left\{ F_i^p - \frac{1}{x_i}\sum_{j=1}^{S_i}F_i^jx_i^j,0 \right\} \nonumber
\EEQA
Note that the total mass of the population $i$ is a constant $x_i$.  An interesting property of BNN dynamics is non-complacency, i.e., it allows extinct strategies to resurface, so that its stationary points are always Wardrop equilibria~\cite{Sandholm01}.

\section{Traffic Splitting: \\Multi-path Network Coding (MPNC) Game}\label{sec:mpncGame}

We model the traffic-splitting process of our Decoupled Dynamics controller as a potential game, $\mathcal{G}$, which we refer to as the \emph{Multi-Path Network Coding Game} (MPNC Game).  Our system model consists of a set of nodes $\mathcal{N}= \{n_1,\dots,n_N\}$, where each node $n_i \in \mathcal{N}$ is surrounded by a random number of other nodes. The cost of transmission (per unit rate) from node $n_i$ to its neighboring node $n_j$ is a constant and is equal to $\alpha_{ij}$, similarly, cost of transmission (per unit rate) from $n_j$ to $n_i$ is $\alpha_{ji}$.  There exists a set of flows (these correspond to players in the game) $\mathcal{F} = \{1,\dots,F\}$.  Each flow, $i \in \mathcal{F}$ is defined as a tuple $(n^{s}_i,n^{d}_i,x_i)$, where  $n^{s}_i \in \mathcal{N}$  is the source node,  $n^{d}_i \in \mathcal{N}$ is the destination node, and $x_i$ {\it packets/sec} is the traffic sent from source to destination.  This traffic is equivalent to the population mass in the population game interpretation.  Each flow $i$ is associated with a set of hyper-paths $\mathcal{S}_i.$
\begin{definition}
A \emph{hyper-path} $p\in \mathcal{S}_i$ between source $n^{s}_i$ and destination $n^{d}_i$ is a virtual path over a physical path between $n^{s}_i$ and $n^{d}_i$. A hyper-path contains zero or more hyper-links on it and at \emph{each node} on the underlying physical path there can be at-most one hyper-link.  It follows that the set of all paths are a subset of the hyper-paths.      
\end{definition}  
In other words, a hyper-path can have a combination of at-most two flows at each node.  A flow can split its traffic among the hyper-paths available to it, and we denote a \emph{sub-flow} $f^p_i$ of flow $i$ by the tuple $(n^{s}_i,n^{d}_i,p,x^p_i)$.  Here, $x^p_i$ is the traffic sent by flow $i$ on hyper-path $p$.  The sum of link costs (per unit rate) on the physical path corresponding to the hyper-path is denoted $\beta_i^p$.  Note that the cost seen by a sub-flow using such a hyper-path might be lower than this cost due to saving attained by network coding.

We represent the division of traffic $x_i$ of flow $i \in \mathcal{F}$, over all the hyper-paths $p \in \mathcal{S}_i$ as a vector, $\vec{x}_i = \{x_i^1,\dots,x_i^{S_i}\}$ such that $\sum_{p=1}^{S_i}x_i^p = x_i$.  $\vec{x}_i$ is called the strategy distribution of flow $i$, and the set of all the strategy distributions of all the flows is called the state of the flows and is represented as  $X = [\vec{x}_1 \dots \vec{x}_F]$. We denote the set of all states of the system as $\mathbb X$, i.e., $X \in \mathbb X$.

The set of all hyper-links in the network is assumed to be pre-determined and is represented by $\mathcal{H} = \{1,\dots,H\}$, where $H$ is the number of hyper-links.  Recall that the hyper-link formed by encoding packets that belong to flows $i$ and $j$, for $i,j \in \mathcal{F}$ at node $n_k$ is represented by $n_k[(i,p,n_i),(j,q,n_j)]$. Nodes $n_i$ and $n_j$ are the next hop nodes for the hyper-path of flows $i$ and $j$, using hyper-paths $p$ and $q$ respectively. Note that we have slightly modified the definition to include the fact that $i$ and $j$ are using hyper-paths.  We denote by $\mathcal{H}_i^p \subseteq \mathcal{H}$ the set of all hyper-links associated with flow $f_i^p$. 

Each hyper-link can choose its \emph{capacity} independently of others.  We denote the capacity of the hyper-link $h=n_k[(i,p,n_i),(j,q,n_j)]$ by $y_h$ {\it packets/sec}.  The hyper-link broadcasts packets received at node $n_k$ to $n_i$ and $n_j$ up to capacity $y_h$.  The vector of all hyper-link capacities is called the hyper-link state and is denoted by, $Y = [y_1,\dots,y_H]$. Let $\mathbb Y$ be the set of all possible hyper-link states, i.e., $Y \in \mathbb Y$. The state of the system is defined as $(X,Y)$, where $X \in \mathbb X$ is the state of the flows and $Y\in \mathbb Y$ is the state of the hyper-links.

In the traffic splitting phase of our algorithm, flows try to attain the state of lowest cost for a given hyper-link 
state $Y$. The hyper-link capacities are controlled in the next phase (hyper-link capacity control), discussed in 
Section~\ref{sec:hlCtrl}

The payoff (per unit rate) obtained in using hyper-path $p \in \mathcal{S}_i$ of flow $i \in \mathcal{F}$ when the state of the system is $(X,Y)$ is denoted by $F_i^p(X,Y) \in \mathbb R$ and is assumed to be continuous and differentiable. We may have to make suitable approximations on cost functions to ensure that these conditions hold.  We model our system as a potential game, using the total cost function $C(X,Y)$ as our potential function. Recall from (\ref{eqn:cost}) that the total cost of the system is
\BEQA\label{eqn:totCost}
C(X,Y) =& \sum_{i=1}^F\sum_{p=1}^{S_i}\beta_i^px_i^p - \sum_{h=1}^H T(h),
\EEQA
where
\BEQA\label{eqn:hlTax}
T(h) &= \alpha_{ki}\min\{x_i^p,y_h\} + \alpha_{kj}\min\{x_j^q,y_h\} \nonumber \\
&- \max\{\alpha_{ki},\alpha_{kj}\} y_h
\EEQA
As can be seen from (\ref{eqn:hlTax}), the cost function contains ``min'' terms over the hyper-link capacity and the flow rates, this makes the function non-continuous and non-differentiable. In order to have a continuously differentiable cost function we approximate these ``min'' terms using a generalized mean-valued function.

Let $a = \{a_1,\dots,a_n\}$ be the set of positive real numbers and let $r$ be some non-zero real number. Then the generalized \emph{r-mean} of $a$ is given by:
\BEQA\label{eqn:rMean}
M_r(a) = \left( \frac{1}{n} \sum_{i=1}^n  a_i^r\right)
\EEQA
The ``min'' function over the set $a$ is approximated using $M_r(a)$ as:
\BEQA\label{eqn:minApprox}
\min\{a_1,\dots,a_n\} = \lim_{r \rightarrow -\infty} M_r(a)
\EEQA
Substituting for $M_r$ (\ref{eqn:rMean}, instead of the ``min'' function in (\ref{eqn:totCost} we get the approximated total cost function as:
\BEQA\label{eqn:Capprox}
\Capprox(X,Y) =& \sum_{i=1}^F\sum_{p=1}^{S_i}\beta_i^px_i^p - \sum_{h=1}^H \tilde{T}(h),
\EEQA
where for a hyper-link $h=n_k[(i,p,n_i),(j,q,n_j)] \in \mathcal{H}$:
\BEQA
\tilde{T}(h) &= \alpha_{ki}\left(\frac{(x_i^p)^r +(y_h)^r}{2}\right)^{\frac{1}{r}} \label{eqn:Tapprox}
+  \alpha_{kj}\left(\frac{(x_i^p)^r +(y_h)^r}{2}\right)^{\frac{1}{r}} \nonumber \\
& - \max\{\alpha_{ki},\alpha_{kj}\} y_h 
\EEQA

The cost function $\Capprox(X,Y)$ is continuous and differentiable. So, we use the approximated cost function as our potential function. Thus, it follows from the definition of potential games (~\ref{defn:potentialGame}) that, the payoff obtained by flow $i \in \mathcal{F}$ in using option $p \in \mathcal{S}_i$ is:
\BEQA 
F_i^p(X,Y) &= \D{\Capprox(X,Y)}{x_i^p} \ \forall i \in \mathcal{F},\ p \in \mathcal{S}_i \label{eqn:F_potential} \\
&= \beta_i^p - \sum_{h \in \mathcal{H}_i^p}\D{\tilde{T}(h)}{x_i^p}, 
\EEQA
where, from (\ref{eqn:Tapprox})
\BEQA
 \D{\tilde{T}(h)}{x_i^p} = \frac{\alpha_{ki}}{2}\left(\frac{x_i^p}{M_r(x_i^p,y_h)}\right)^{r-1} 
\EEQA
Recall that
\BEQA
M_r(x_i^p,y_h) = \left(\frac{(x_i^p)^r +(y_h)^r}{2}\right)^{\frac{1}{r}}.
\EEQA
Hence,
\BEQA\label{eqn:F_i_j}
F_i^p(X,Y) &= \beta_i^p - \sum_{h \in \mathcal{H}_i^p}\frac{\alpha_{ki}}{2}\left(\frac{x_i^p}{M_r(x_i^p,y_h)}\right)^{r-1},
\EEQA
where $\mathcal{H}_i^p$ is the set of all hyper-links associated with sub-flow $f_i^p$. Note, the payoff is the cost incurred in using an option, so the players try to minimize their cost.
The source node of each flow, $i \in \mathcal{F}$, observes the marginal cost, $F_i^p$, obtained in using a particular option, $p \in \mathcal{S}_i$, and changes the mass on that particular option,$x_i^p$, so as to attain Wardrop equilibrium~\cite{Wardrop52}. The source nodes use BNN dynamics (\ref{eqn:origbnn}) to control the mass on each option. But since each source tries to \emph{minimize} its payoff, we use a modified version BNN dynamics:
\BEQA\label{eqn:bnn}
\dot{x}_f^p  = \left(x_f\gamma_f^p - x^p_f\sum_{j=1}^{S_f}\gamma_f^j \right),\\ 
\mbox{where, }\gamma_f^p = \max\left\{\frac{1}{x_f}\sum_{j=1}^{S_f}F_f^jx_f^j - F_f^p ,0 \right\} \nonumber
\EEQA
In  the next section, we prove the stability of our system using Lyapunov theory.



\subsection{Convergence of MPNC Game}

We show in this section that the multi-path network coding game converges to a stationary point when each source uses BNN dynamics.  We will use the theory of Lyapunov functions~\cite{khalil96} to show that our population game $\mathcal{G}$, is stable for a given hyper-link state $\hat{Y}$.  We use the approximated total cost of the system~(\ref{eqn:Capprox}) as our Lyapunov function.
\begin{theorem} \label{thm:bnnStable}
The system of flows $\mathcal{F}$ that use \emph{BNN dynamics} with payoffs given by (\ref{eqn:F_i_j}) is globally asymptotically stable for a given hyper-link state $\hat{Y}$.
\end{theorem}
\begin{proof}
\newcommand{\XYstate}{\tilde{\mathcal{X}}_{\hat{Y}}}
\newcommand{\Lfn}{\mathcal{L}_{\hat{Y}}}
\newcommand{\Ldot}{\dot{\mathcal{L}}_{\hat{Y}}}
We use the approximated total cost function $\Capprox(X,Y)$ (\ref{eqn:Capprox}) as our Lyapunov function.  It is simple to verify that the cost function $\Capprox(X,\hat{Y})$, is non-negative is convex, and hence is a valid candidate.   For a given hyper-link state, $\hat{Y}$, we define our Lyapunov function as:
\BEQA\label{eqn:Lyap}
\Lfn(X) = \Capprox(X,\hat{Y})
\EEQA
From (\ref{eqn:F_potential})
\BEQA
\D{\Lfn(X)}{x_f^p} = \D{\Capprox(X,\hat{Y})}{x_f^p} = F_f^p(X,\hat{Y}).
\EEQA Hence,
\BEQA
\Ldot(X) &= \sum_{f=1}^F \sum_{p=1}^{S_f} \D{\Lfn(X)}{x_f^p}\dot{x}_f^p \\
&= \sum_{f=1}^F \sum_{p=1}^{S_f} F_f^p(X,\hat{Y})\dot{x}_f^p
\EEQA
%
From (\ref{eqn:bnn}) we can substitute the value for $\dot{x}_f^p$ and we have
\BEQA \nonumber
\Ldot(X) = \sum_{f=1}^F \sum_{p=1}^{S_f} F_f^p (x_f\gamma_f^p - x^p_f\sum_{j=1}^{S_f}\gamma_f^j) \\
= \sum_{f=1}^{F}x_f\left(\sum_{p=1}^{S_f}F_f^P\gamma_f^p - \left(\frac{1}{x_f}\sum_{p=1}^{S_f}F_f^px_f^p\right)\sum_{j=1}^{S_f}\gamma_f^j\right)
\EEQA
We define
\BEQA 
\bar{F}_f \triangleq \frac{1}{x_f}\sum_{p=1}^{S_f}F_f^px_f^p \nonumber \hspace{1in}\\
\implies \sum_{f=1}^{F}x_f\left(\sum_{p=1}^{S_f}F_f^P\gamma_f^p - \sum_{j=1}^{S_f}\bar{F}_f\gamma_f^j\right) \\
= \sum_{f=1}^{F}x_f\left(\sum_{p=1}^{S_f}\gamma_f^p (F_f^P - \bar{F}_f)\right) \quad  \quad  \quad\\
\leq - \sum_{f=1}^Fx_f\left(\sum_{p=1}^{S_f}(\gamma_f^p)^2\right) \leq 0 \hspace{0.5in}
\EEQA
Thus,
\BEQA
\Ldot(X) & \leq & 0, \quad \forall\  X \in \mathcal{X}
\EEQA
where equality exists when the state $X$ corresponds to the stationary point of BNN dynamics.
Hence, the system is globally asymptotically stable.
\end{proof}

\subsection{Efficiency}

 The objective of our system is to minimize the total cost for a given load vector $\vec{x} = [x_1, \dots, x_Q]$ and given hyper-link state $\hat{Y}$. Here the total cost in the system is $\Capprox(X,\hat{Y})$ and is defined in (\ref{eqn:totCost}). This can be represented as the following constrained minimization problem:
\BEQA\label{eqn:sysPrimal}
\min_{X} \Capprox(X,\hat{Y}) \\
\mbox{subject to: }\nonumber\\
\sum_{p=1}^{S_i}x_i^p &=& x_i \quad \forall \ i \in \mathcal{F}\label{eqn_sysConstraints} \\
 x_i^p &\geq & 0. \nonumber
\EEQA
The Lagrange dual associated with the above minimization problem, for a given $\hat{Y}$ is
\newcommand{\Lagr}{\mathcal{L}_{\hat{Y}}}
\BEQA\label{eqn:sysDual}
\Lagr(\lambda,h,X) = \max_{\lambda, h}\min_{X}\bigg(\Capprox(X,\hat{Y}) &-& \\
\sum_{i=1}^F \lambda _i\Big(\sum_{p=1}^{S_i}x_i^p -x_i\Big) &-& \sum_{i=1}^F \sum_{p=1}^{S_i} h_i^p x_i^p \bigg) \nonumber
\EEQA
where $\lambda_i$ and $h_p^i \geq 0$ , $\forall \ i \in \mathcal{F}$ and $p \in \mathcal{S}_i$, are the dual variables.  Now the above dual problem gives the following Karush-Kuhn-Tucker first order conditions:
\BEQA
\D{\Lagr}{x_i^p}(\lambda,h, X^\star) = 0 &  \forall \ i \in \mathcal{F} \mbox{ and }p \in \mathcal{S}_i  \label{eqn:sysKKT} \\
\mbox{and} \nonumber \\ 
h_i^p x_i^{\star p} = 0 & \forall\  i \in \mathcal{F}\mbox{ and }p \in \mathcal{S}_i\label{eqn:KKT2}
\EEQA
where $X^\star$ is the global minimum for the primal problem (\ref{eqn:sysPrimal}).  Hence from (\ref{eqn:sysKKT}) we have, 
$\forall\ i \in \mathcal{F}$ and $\forall\ p \in \mathcal{S}_i$,
\BEQA \nonumber
\D{\Capprox}{x_i^p}(X^\star,\hat{Y}) - \lambda_i \D{(\sum_{p=1}^{S_i} x^{\star p}_i - x^\star i)}{x_i^p} + h_i^p\ = \ 0 \\
\Rightarrow \D{\Capprox}{x_i^p}(X^\star,\hat{Y}) \ = \ \lambda_i + h_i^p \label{eqn:KKT} \\
\Rightarrow F_i^p(X^\star,\hat{Y})\ = \ \lambda_i + h_i^p
\EEQA
where the last equation follows from the definition of potential games (\ref{defn:potentialGame}).\\
From (\ref{eqn:KKT2}), it follows that 
\BEQA
F_i^p(X^\star,\hat{Y})\ = \ \lambda_i & \mbox{ when } x_i^{\star p} > 0 \label{eqn:wdropEq1}\\
\mbox{and} \nonumber \\
F_i^p(X^\star,\hat{Y})  \ = \ \lambda_i + h_i^p & \mbox{ when } x_i^{\star p} = 0 \label{eqn:wdropEq2}
\EEQA
$\forall\ i \in \mathcal{F}$ and $\forall\ p \in \mathcal{S}_i$. The above condition 
(\ref{eqn:wdropEq1}, \ref{eqn:wdropEq2}), implies that the payoff on all the options used is identical and for options not in use the payoff is more, which is equivalent to the definition of Wardrop equilibrium (\ref{defn:wardrop}). Notice, we use a modified definition of Wardrop equilibrium, since each source tries to minimize it's cost (or payoff).

We have the following theorem that proves the efficiency of our system.
\begin{theorem}
The solution of the minimization problem in (\ref{eqn:sysPrimal}) is identical to the Wardrop equilibrium of the non-cooperative potential game $\mathcal{G}$.
\end{theorem}
\begin{proof}
Consider the BNN dynamics (\ref{eqn:bnn}), at stationary point, $\tilde{X}$, we have $\dot{x}_i^p = 0$, which implies that either,
\BEQA
&\hat{F}_i = \  F_i^p(\tilde{X},\hat{Y}) \label{eqn_statPoint} \\
\mbox{or } &\hat{x}_i^p = 0,\nonumber \\
\mbox{where, }&\hat{F}_i \triangleq \frac{1}{\hat{x}_i}\sum_{r=1}^Q\hat{x}_i^r F_i^r(\tilde{X},\hat{Y}) \quad \forall\ i \in \mathcal{F}, \label{eqn_Favg}
\EEQA
The above expressions imply that, all the hyper-paths used by a particular flow, $i \in \mathcal{F}$, will yield same payoff, $\hat{F}_i$, while hyper-paths not used ($x_i^p = 0$) would yield a payoff higher than $\hat{F}_i$.

We observe that the conditions required for Wardrop equilibrium are identical to the KKT first order conditions (\ref{eqn:wdropEq1})-(\ref{eqn:wdropEq2}) of the minimization problem (\ref{eqn:sysPrimal}) when
 \BEQA
\hat{F}_i = \lambda_i \quad \forall\ i \in \mathcal{F} \nonumber
\EEQA
It follows from the convexity of the total system cost that, there is no duality-gap between the primal (\ref{eqn:sysPrimal}) and the dual (\ref{eqn:sysDual}) problems. Thus, the optimal primal solution is equal to optimal dual solutions, which is identical to the Wardrop equilibrium.
\end{proof}
\section{Hyper-Link Capacity Control}\label{sec:hlCtrl}
Thus far we have designed a distributed scheme that would result in minimum cost for a given hyper-link state or capacities $Y$ and for a given load vector $\vec{x} = \{x_1,\dots,x_f\}$. 
In this phase of Decoupled Dynamics, the hyper-link capacities are adjusted based on the current system cost so as to guarantee a minimum total system cost for a given load vector $\vec{x}$. This phase runs at a larger time-scale as compared to the traffic splitting phase described in Section~\ref{sec:mpncGame}. It is assumed that during this phase all the flows or players remain in equilibrium, i.e., changing the hyper-link capacities would force all the source nodes to attain Wardrop equilibrium instantaneously. 

The hyper-link capacity control can be formulated as a centralized convex optimization problem as follows:
\BEQA 
\min_{Y}& H(Y)\label{eqn:hlCtrl_CnvxOpt}\\
\mbox{subject to, } & y_h \geq 0 \ \forall y_h \in Y \mbox{ and }\forall h \in \mathcal{H}\nonumber
\EEQA
where, $H(Y)$ is the minimum total cost of the system for a given hyper-link state $Y$, i.e., $H(Y) = \Capprox(X^\star,Y)$, where, for a given $Y$, $X^\star$ is an optimal state of the flows that results in minimum cost.\footnote{Notice, there could be many different states, $X^\star$, which result in a minimum cost but the minimum value, $\Capprox(X^\star,Y)$, is unique.}
We use a simple gradient controller defined below:
\BEQA\label{eqn:HLControl}
\dot{y}_h = -\kappa \D{H(Y)}{y_h}\ \forall y_h \in Y 
\EEQA
The partial derivative, $\D{H(Y)}{y_h}$, is over the variables $y_h \in Y$. Changing the hyper-link capacity $y_h$, of some hyper-link $h\in \mathcal{H}$, would result in a different state of the flows, $X_h^\star$ and hence a different minimum cost, $C(X_h^\star,Y_h)$, where $Y_h$ corresponds to the changed hyper-link capacity of $y_h$ while other capacities are fixed, as compared to $Y$. Thus for a hyper-link, $h=n_k[(f_i^p,n_i),(f_j^q,n_j)]$ with capacity $y_h$,
\BEQA
\D{H(Y)}{y_h} &= \D{\Capprox}{y_h}(X^\star,Y) + \sum_{i=1}^F\sum_{p=1}^{S_i} \D{\Capprox}{x_i^p}(X^\star,Y)\D{x_i^p}{y_h} \\
&= \D{\Capprox}{y_h}(X^\star,Y) + \sum_{i=1}^F F_i\sum_{p=1}^{S_i}\D{x_i^p}{y_h} \nonumber
\EEQA
where, the last expression follows from the definition of $F_i^p$ (\ref{eqn:F_potential}) and the fact that, for changes in the hyper-link state, the sources attain Wardrop equilibrium instantaneously. In other words, before and after a small change in $y_h$ the system is in Wardrop equilibrium. Hence, $F_i^p = F_i\ \forall i \in \mathcal{F}$ and $\forall p \in \mathcal{S}_i$.
Finally, $\sum_{p=1}^{S_i}\D{x_i^p}{y_h} = 0$ since, the total load $x_i = \sum_{p=1}^{S_i}x_i^p$ is fixed.
\BEQA\label{eqn:ctrlPrimal}
\Rightarrow \D{H(Y)}{y_h} &= \D{\Capprox}{y_h}(X^\star,Y) = -\D{\tilde{T}}{y_h}(h)
\EEQA
where, from (\ref{eqn:Tapprox}), for hyper-link $h=n_k[(f_i^p,n_i),(f_j^q,n_j)]$,
\BEQA
\D{\tilde{T}}{y_h}(h) &=\frac{\alpha_{ki}}{2}\left(\frac{y_h}{M_r(x_i^p,y_h)}\right)^{r-1} + \frac{\alpha_{kj}}{2}\left(\frac{y_h}{M_r(x_j^q,y_h)}\right)^{r-1} \nonumber\\
&- \max\{\alpha_{ki},\alpha_{kj} \} \nonumber\\
\mbox{Recall, }& M_r(x_i^p,y_h) = \left(\frac{(x_i^p)^r +(y_h)^r}{2}\right)^{\frac{1}{r}} \nonumber
\EEQA

\begin{theorem}
At the large time-scale, the hyper-link capacity control with dynamics (\ref{eqn:ctrlPrimal}) is globally asymptotically stable.
\end{theorem}
\begin{proof}
We use the following Lyapunov function
\vspace{-0.1in}
\BEQA
 Z(Y) &=&  V(Y)- V(\hat{Y}) \label{eqn:admCtrlLyap} \\
\mbox{where } V(Y)&=& \kappa H(Y) \label{eqn:admCtrlV}
\EEQA
which is strictly convex, with $\hat{Y}$ is the hyper-link state which results in minimum cost $H(Y)$.  Differentiating $Z(Y)$ we obtain
\vspace{-0.2in}
\BEQA\label{eqn:Zdot}
\dot{Z} = \sum_{h=1}^H \D {V}{y_h}\dot{y_h}.
\EEQA
Then from (\ref{eqn:admCtrlV}) and (\ref{eqn:ctrlPrimal})
\BEQA
\D {V}{y_h} = \kappa\D{H(Y)}{y_h}\dot{y}_h &= \frac{\dot{y}_h}{y_h} \\
\therefore\ \dot{Z} = \sum_{h=1}^F \frac{\dot{y}_h^2}{y_h} \leq 0 & \forall\ Y,
\EEQA
with $\dot{Z} = 0$ at the stationary points of the system. Thus, the system is globally asymptotically stable \cite{khalil96}.
\end{proof}

Finally, it's not hard to show that the equilibrium conditions of the controller (\ref{eqn:ctrlPrimal}) are the same as the KKT conditions of the optimization problem (\ref{eqn:hlCtrl_CnvxOpt}).  Hence, the controller succeeds in minimizing the total cost of the system for a given load $\vec{x}$ into the system.  Thus, the system state converges to a local minimum.  Now, since the global cost minimization problem under the $\min$ approximation is convex, the solution is the global minimum of the relaxed problem. 

\section{Simulations}
We simulated our system in Matlab to show system convergence.  We first performed our simulations for our simple network shown in Figure~\ref{fig:codingOpportunities}. The load at the source nodes $1$, $2$ and $3$ is given as $4.73$, $2.69$ and $3.56$ respectively. We use the following costs on the individual links ($\alpha_{ij}$):
 $\alpha_{12} = 2.8$, $\alpha_{23} = 1.6$, $\alpha_{34} = 1.8$, $\alpha_{25} = 1.3$, $\alpha_{54}= 2.1$, $\alpha_{26} = 1.7$, $\alpha_{48} = 2.9$, $\alpha_{86} = 2.2$, $\alpha_{57} = 1.9$, $\alpha_{71}= 2.6$; we assume the costs on the links are symmetric. We use the approximated cost function (\ref{eqn:Capprox}), with a value of $r = -100$  for the approximation parameter (\ref{eqn:minApprox}) for our simulations. The simulation is run for $50$ large time scale units, and in each large time scale we have $20$ small time units. 

We compare the total cost of the system for the following scenarios:
\begin{enumerate}
\item Decoupled Dynamics: This is the algorithm that we developed; we use our hyper-links to decouple the flows that participate in coding.
\item Coupled Dynamics (no hyper-link): Here, there is coupling between individual flows and coding happens at the minimum rate of the constituent flows. We use similar game dynamics as that was used in DD. The total cost is specified in Equation~(\ref{equ:energy}).
\item No Coding: In this system no network coding is used.
\item LP Optimal: This is a centralized solution. We formulated our system as a Linear Program (LP)  of minimizing cost (\ref{eqn:totCost}) over $X$ and $Y$ for a given load vector  that we obtain using an LP-solver. 
\end{enumerate} 

As seen in the Figure~\ref{fig:costComparison}, the total cost of the system (number of transmissions per unit rate) for our model (decouples using hyper-link) is close to the optimal solution obtained by solving it in a centralized fashion.
\begin{figure}[htbp]
\vspace{-0.1in}
\begin{center}
\includegraphics[width=3in]{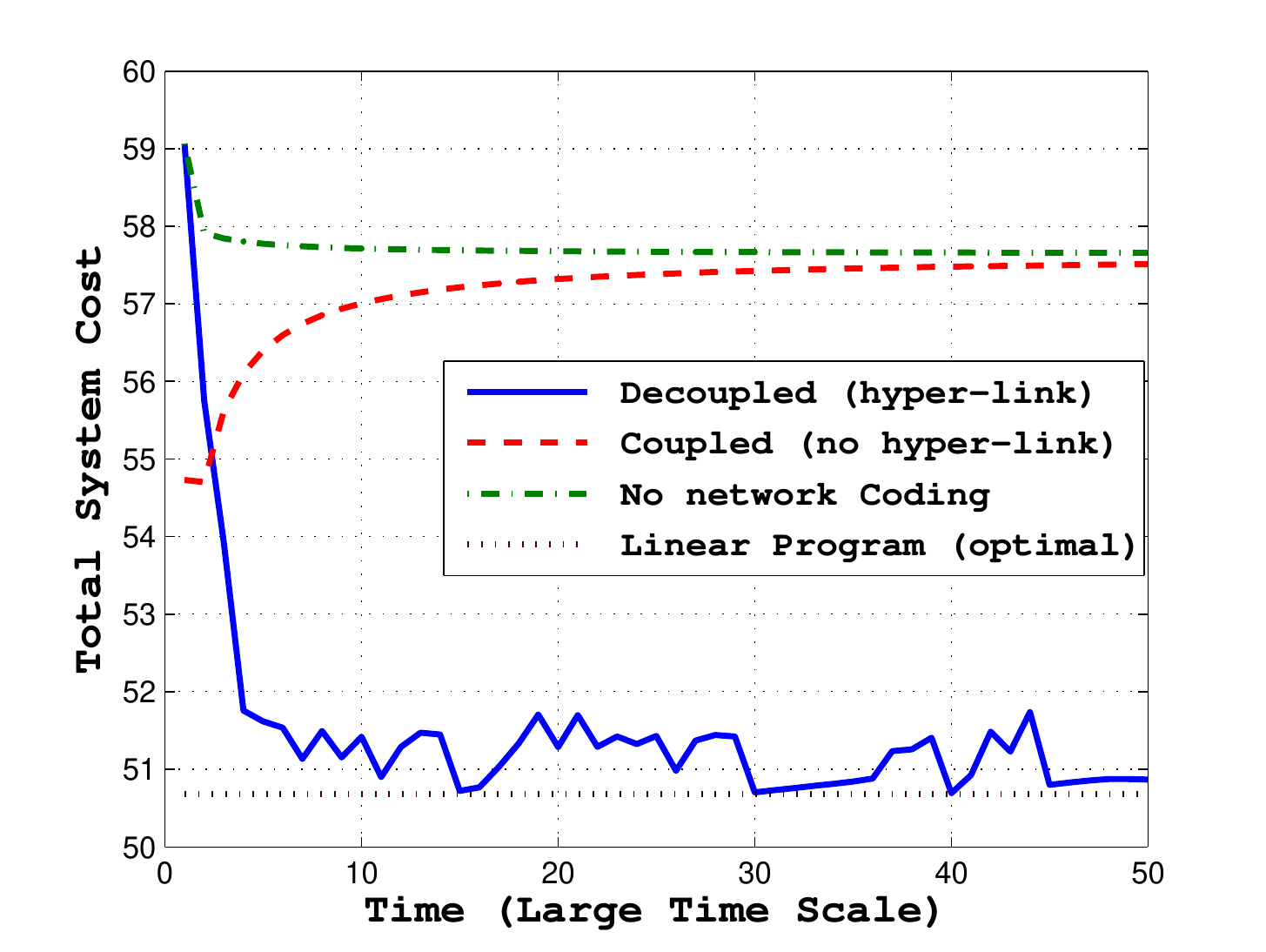}
\figCaptionSpace
\caption{Comparison of total system cost (per unit rate), for different systems: DD, CD and non-coded against LP.} 
\label{fig:costComparison}
\end{center}
\vspace{-0.2in}
\end{figure}

For this simple network we compared the final system state of DD and CD with that of the solution obtained using LP. We observe from  Table \ref{tab:bigTopoCostComprison} that the values for the split ($X$) and the hyper-link capacities ($Y$) generated by DD are near-optimal, but CD is very different.
\begin{table}
\begin{center}
\begin{tabular}{||c|c|c|c|c|c|c|c|c||}
\hline
Variable & $x_1^1$ & $x_1^2$ & $x_2^1$ & $x_2^2$ & $x_3^1$ & $x_3^2$ & $y_2$ & $y_3$ \\
\hline
LP  & 2.69 & 2.04 & 2.69 & 0.00 & 0.00 & 3.56 & 3.56 & 2.69 \\
DD  & 2.37 & 3.35 & 2.67 & 0.01 & 0.02 & 3.54 & 3.24 & 2.49 \\
 CD  & 4.70 & 0.02 & 0.07 & 2.61 & 0.03 & 3.52 &  NA  & NA  \\
\hline
\end{tabular}
\caption{Comparison of state variables for LP and DD and CD.}
\label{tab:resComparison}
\end{center}
\vspace{-0.5in}
\end{table}

Next, we perform our simulations on a bigger topology shown in Figure~\ref{fig:bigTopo}. This network consists of $30$ nodes shared by 6 flows. Flows $1$, $2$, $3$ and $6$ have two hyper-paths and flows $4$ and $5$ have three hyper-paths. 
There are $6$ hyper-links in the system.
\begin{figure}[t]
\vspace{-0.2in}
\begin{center}
\includegraphics[width=2.9in]{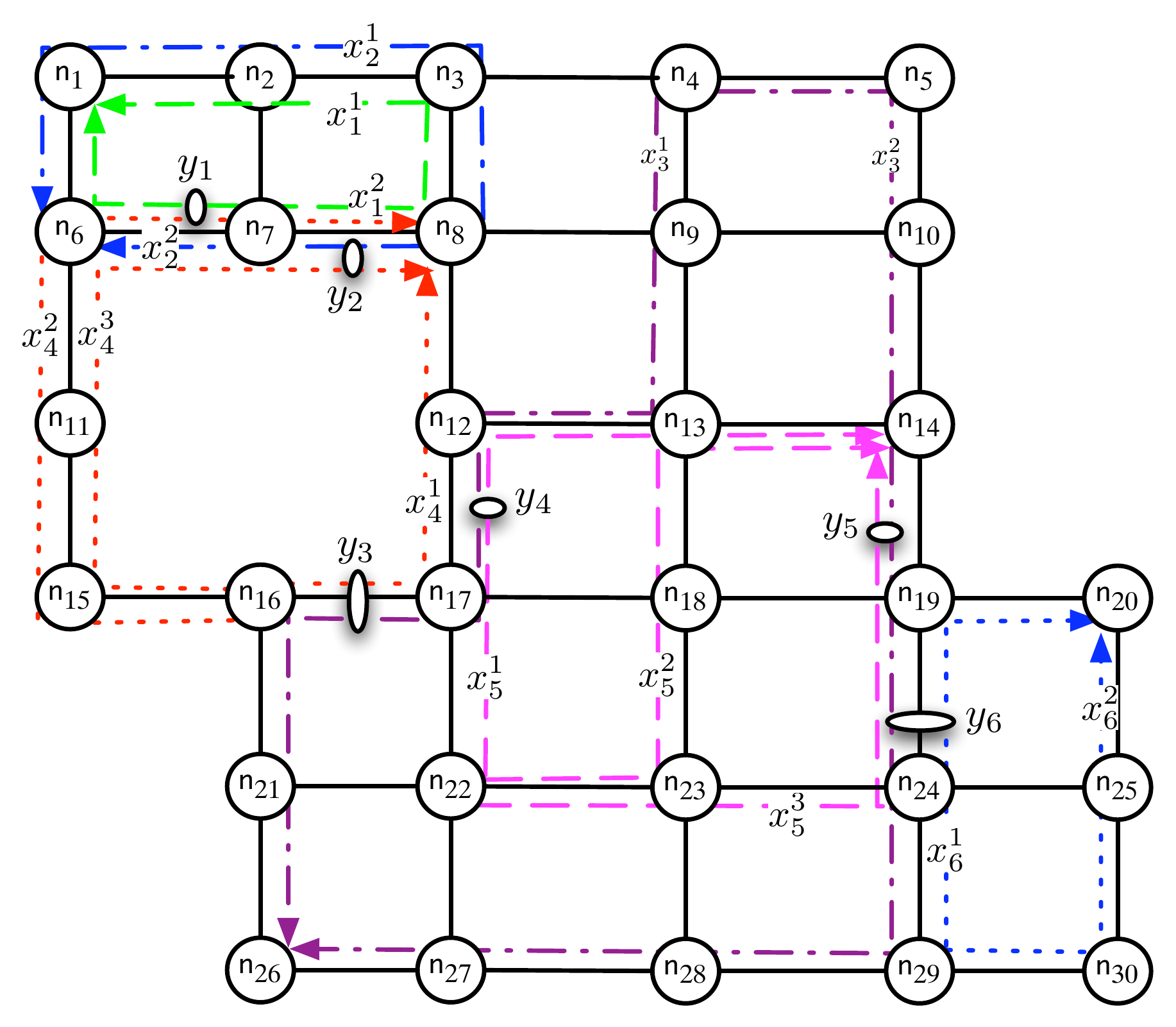}
\figCaptionSpace
\caption{Network topology.}
\label{fig:bigTopo}
\end{center}
\figSpace
\end{figure}
Our results are shown in Table~\ref{tab:bigTopoCostComprison}. We observe that DD performs near-optimally and significantly ouperforms CD in terms of total cost.
\begin{table}
\begin{center}
\begin{tabular}{||c|c|c|c||}
\hline
case & LP & DD & CD \\
\hline
1 &	1293.3437 & 1298.3194 &	1325.0618 \\
2 &	1550.4593 & 1563.7340 &	1625.06315 \\
3 &	1624.6393 &	1638.4021 & 1642.8801 \\
4 &	1826.8595 &	1837.1489 & 1865.6998 \\
\hline
\end{tabular}
\caption{Total system cost comparison of Decoupled dynamics(DD) and Coupled dynamics (CD) against the LP solution.}
\label{tab:bigTopoCostComprison}
\end{center}
\vspace{-0.4in}
\end{table}

\section{Conclusion}

We consider a wireless network with given costs on arcs,  
traffic matrix and multiple paths. The objective is to find the splits of traffic for each source across its multiple paths in
a distributed manner leveraging the reverse carpooling technique. For 
this we relax to problem into two sub-problems, and propose a two-level distributed control scheme set up as a game 
between the sources and the hyperlink nodes. On one level, 
given a set of hyperlink capacities, the sources selfishly choose their 
splits and attain a Wardrop equilibrium. On the other level, given the
traffic splits, the hyperlinks may slightly increase or decrease their
capacities using a steepest descent algorithm. We construct a Lyapunov 
function argument to show that this process asymptotically converges, although performed selfishly in a distributed fashion.

We performed several numerical studies and found that our two-level 
controller converges fast to the optimal solutions. 
Some of the bi-products of 
our experiments were that: more expensive paths before 
network coding became cheaper and shortest paths were not necessarily
optimal. In conclusion, from a methodological standpoint we have a
distributed controller that achieves a near-optimal solution when the
individuals are self-interested. 
\bibliographystyle{IEEEtran}
\bibliography{IEEEabrv,netCode}  

%
%

\end{document}